\documentclass[preprint,12pt]{article}
\usepackage[T1]{fontenc}
\usepackage[utf8]{inputenc}
\usepackage{amsmath,amsfonts,amssymb,amsthm}
\usepackage{tikz}
\usepackage{xspace}
\usepackage{todonotes}
\usepackage{framed}
\usepackage{enumerate}
\usetikzlibrary{calc,positioning}

\renewcommand{\emph}[1]{{\textit{ #1}}}

\newtheorem{theorem}{Theorem}
\newtheorem{lemma}[theorem]{Lemma}
\newtheorem{clm}{Claim}
\newtheorem{corollary}[theorem]{Corollary}
\newtheorem{proposition}[theorem]{Proposition}
\theoremstyle{definition}
\newtheorem{problem}{Problem}

\def\pn3sat{\textsc{Pure-Nae-3-Sat}\xspace}

\def\NN{\mathbb{N}}

\def\L{{\begin{tikzpicture}[scale=0.2]\draw (0,0) -- (0,-1) -- (1,-1);
\end{tikzpicture}}}
\def\RL{{\begin{tikzpicture}[scale=0.2]\draw (0,-1) -- (1,-1) -- (1,0);
\end{tikzpicture}}}
\def\G{{\begin{tikzpicture}[scale=0.2]\draw (0,0) -- (0,1) -- (1,1);
\end{tikzpicture}}}
\def\RG{{\begin{tikzpicture}[scale=0.2]\draw (0,0) -- (1,0) -- (1,-1);
\end{tikzpicture}}}
\def\U{{\begin{tikzpicture}[scale=0.2]\draw (0,0) -- (0,-1) -- (1,-1) -- (1,0);\end{tikzpicture}}}
\def\RU{{\begin{tikzpicture}[scale=0.2]\draw  (0,0) -- (0,1)--  (1,1) -- (1,0);
\end{tikzpicture}}}
\def\C{{\begin{tikzpicture}[scale=0.2]\draw (1,1) -- (0,1) -- (0,0) -- (1,0);\end{tikzpicture}}}
\def\RC{{\begin{tikzpicture}[scale=0.2]\draw (0,0) -- (1,0) --(1,1) -- (0,1);\end{tikzpicture}}}
\def\Z{{\begin{tikzpicture}[scale=0.2]\draw (0,0) -- (1,0) -- (1,-0.5) -- (2,-0.5);\end{tikzpicture}}}
\def\RZ{{\begin{tikzpicture}[scale=0.2]\draw (0,-0.5) -- (1,-0.5) -- (1,0) -- (2,0);\end{tikzpicture}}}
\def\N{{\begin{tikzpicture}[scale=0.2]\draw (0,-1) -- (0,-0.5) -- (1,-0.5) -- (1,0);\end{tikzpicture}}}
\def\RN{{\begin{tikzpicture}[scale=0.2]\draw (0,0) -- (0,-0.5) -- (1,-0.5) -- (1,-1);\end{tikzpicture}}}
\def\O{{\begin{tikzpicture}[scale=0.2]\draw (0,-0.2) -- (0,-1) -- (1,-1) -- (1,0) -- (0.2,0);\end{tikzpicture}}}

\input epsf 

\begin{document}

\title{On edge intersection graphs of paths with 2 bends\footnote{The extended abstract of this paper was presented on the conference WG 2016 \cite{WG}.}}

\author{Martin Pergel\footnote{Department of Software and Computer Science Education,
Charles University, Praha, Czech Republic. Partially supported by a Czech research grant GA\v{C}R GA14-10799S.}
\and
Pawe{\l} Rz\k{a}\.{z}ewski\footnote{Faculty of Mathematics and Information Science, Warsaw University of Technology, Warsaw, Poland} \footnote{Institute of Computer Science and Control, Hungarian Academy of Sciences (MTA SZTAKI), Budapest, Hungary. Supported by ERC Starting Grant PARAMTIGHT (No. 280152).}
}
\date{}
\maketitle

\begin{abstract}
An EPG-representation of a graph $G$ is a collection of paths in a plane square grid, each  corresponding to a single vertex of $G$, so that two vertices are adjacent if and only if their corresponding paths share infinitely many points.
In this paper we focus on graphs admitting EPG-representations by paths with at most 2 bends. We show hardness of the recognition problem for this class of graphs, along with some subclasses.

We also initiate the study of graphs representable by unaligned polylines, and by polylines, whose every segment is parallel to one of prescribed slopes. We show hardness of recognition and explore the trade-off between the number of bends and the number of slopes.
\end{abstract}

\section{Introduction}
The concept of {\em edge intersection graphs of paths in a grid} ({\em EPG-graphs}) was introduced by Golumbic {\em et al.} \cite{Izraelci}.
By an EPG-representation of a graph $G$ we mean a mapping from vertices of $G$ to paths in a planar square grid, such that two vertices are adjacent if and only if their corresponding paths share a grid edge.
As each graph can be represented in this way (see Golumbic {\em et al.} \cite{Izraelci}),
it makes sense to consider representations with some restricted sets of shapes. A usual parameterization is by bounding the number $k$ of times each path is allowed to change the direction. Graphs with such a representation are called {\em $k$-bend graphs}.
There are two main branches in this kind of research. The first one is understanding the structure of graphs with at most $k$ bends  -- so far, the case of 1-bend graphs received most attention \cite{Izraelci,Steve, EGM, AR}. The other is finding the smallest $k$, such that every graph of a given class $\cal G$ is a $k$-bend graph. The most interesting results seem to concern planar graphs \cite{BS,Nemci2}.

Since 0-bend graphs are just {\em interval graphs}, they can be recognized in  polynomial time (see e.g. Booth and Lueker \cite{BL}). The recognition of 1-bend graphs is NP-complete (see Heldt {\em et al.} \cite{Nemci}), even if the representation is restricted to any prescribed set of 1-bend objects (see Cameron {\em et al.} \cite{Steve}).
However, the problem becomes trivially solvable when $k$ is at least the maximum degree of the input graph \cite{Nemci}. Thus it is unclear whether $k$-bend graphs are hard to recognize for all $k \geq 2$.

It is worth mentioning the closely related notion of {\em $B_k$-VPG-graphs}. These graphs are defined as intersection graphs of axis-aligned paths with at most $k$ bends. So, unlike in the EPG-representation, paths that share a finite number of points define adjacent vertices. Chaplick {\em et al.} \cite{CJKV} showed it is NP-complete to recognize $B_k$-VPG-graphs, for all $k \geq 0$.

In this paper we explore the problem of recognition of subclasses of EPG-graphs. Namely, we show that it is NP-complete to recognize 2-bend graphs.
We also consider some restrictions, where we permit just some
types of the curves in an EPG-representation (similarly to \cite{Steve}).
One of these restrictions, i.e., {\em monotonic EPG-representations}, where each path ascends in rows and columns, was already considered by Golumbic {\em et al.} \cite{Izraelci}. Our hardness proof even shows that
between monotonic 2-bend graphs and 2-bend graphs, no polynomially
recognizable class can be found.

The class of 2-bend graphs can be perceived as a generalization of the quite well-studied class of 1-bend graphs. We also consider some generalizations of the concept of EPG-representations. We do not require individual segments to be axis-aligned, but
we permit them to use any slope. We call such graphs {\em unaligned EPG-graphs} and study the number of bends in this setting.
After this generalization, we may ask about particular restrictions. These restrictions are represented by restricting number of slopes that segments may use
or even by using just prescribed shapes (in a flavor similar to \cite{Steve}). 

For unaligned EPG-graphs, we show that it is NP-hard to determine
whether a graph is an unaligned 2-bend graph (hardness of the recognition
for 1-bend graphs follows from \cite{Steve}).

Having introduced unaligned EPG-graphs, we observe that there is some trade-off between the number of bends and the number of slopes used in a representation.
We also show that given an unaligned 2-bend graph on $n$ vertices,
we may need $\Omega(\sqrt{n})$ slopes to represent it. This result
appears to be a corollary of our hardness-reduction.

The paper is organized as follows. We start with some definitions and preliminary observations on the structure of 2-bend graphs. In Section \ref{sec:aligned} we prove that it is NP-complete to recognize the graphs from this class and also from some of its subclasses.
In Section \ref{sec:unaligned} we introduce unaligned EPG-graphs and show NP-hardness of recognition of unaligned 2-bend graphs. Then we show a lower bound on the number of slopes required for the representation of any unaligned 2-bend graph on $n$ vertices and we discuss the relations between the number of bends and the number of slopes. The paper is concluded with some open questions in Section \ref{sec:conclusion}.

\section{Preliminaries} \label{sec:defs}
For an EPG-representation of a graph $G$, by $P_v$ we shall denote the path representing a vertex $v$. Often we shall identify the vertex $v$ with the path $P_v$. For example, if we say that two paths are {\em adjacent}, we mean that they share infinitely many points. Note that if two paths {\em intersect}, one common point is enough.

A central notion in the study of EPG-graphs is the {\em bend number}. The bend number of a graph $G$, denoted by $b(G)$, is the minimum $k$, such that $G$ has an EPG-representation, in which every path changes it direction at most $k$ times.
Without loss of generality we can assume that every path in a $k$-bend EPG-representation bends exactly $k$ times \cite{Steve}.

Each 2-bend path will be classified as {\em vertical} or {\em horizontal}, if its middle segment is respectively vertical or horizontal. This middle segment will be called the {\em body} of the path, while the remaining two segments will be referenced as its {\em legs}. 

For a set $X$ of shapes of polylines (i.e., piecewise-linear curves), by $X$-graphs we shall denote the class of graphs admitting an EPG-representation, in which the shape of every path is in $X$ (similar notation was used in \cite{Steve}). So for example 1-bend graphs are $\{\L,\RL,\G,\RG\}$-graphs, while monotonic 2-bend graphs are exactly $\{\RZ,\N\}$-graphs.

Golumbic {\em et al.} \cite{Izraelci} analyzed the structure of cliques in 1-bend graphs and proved that in 1-bend graphs each clique $C$ is either an {\em edge-clique} or a {\em claw-clique}. 
A maximal edge-clique consists of vertices whose representing paths share a common grid edge. A {\em claw} is a set of three distinct grid edges sharing a single endpoint and a maximal claw-clique consists of all paths containing two out of three edges of a given claw. Since we can safely assume that each 1-bend representation of a graph with $n$ vertices can be embedded in a $2n \times 2n$ grid (see also \cite{Izraelci,AS}), we obtain that the number of maximal cliques in a 1-bend graph is at most $O(n^2)$, i.e., is polynomial in $n$. This is no longer the case with 2-bend graphs.

Let $n$ be a positive integer and let $K_{2n}^-$ be the {\em cocktail-party graph}, i.e., a complete graph on $2n$ vertices with a perfect matching removed. It is clear that $K_{2n}^-$ has $2^n = 2^{|V(K_{2n}^-)|/2}$ maximal cliques. Fig. \ref{fig-cliques} (left) shows that $K_{2n}^-$ is a 2-bend graph. Thus we obtain the following.

\begin{proposition}
2-bend graphs can have exponentially many maximal cliques.
\end{proposition}

The restricted structure of cliques in 1-bend graphs follows from the fact that the 1-bend paths representing pairwise adjacent vertices must all share at least one grid point. It is easy to observe that cliques in 2-bend graphs do not have such a simple structure. One could be inclined by Fig. \ref{fig-cliques} (left) that every maximal clique is contained in the union of two edge-cliques or claw-cliques (a similar situation appears in unit disk graphs and is the main ingredient of a polynomial algorithm for {\sc Clique} in these graphs -- see Clark {\em et al.} \cite{CCJ}). However, Fig. \ref{fig-cliques} (right) shows it is not true.
\begin{figure}[h]
\centering
\begin{tikzpicture}[scale=0.6]
\def\n{4}
\def\s{8}
\def\h{2}
\def\w{2.5}
\foreach \i in {0,...,\n}
{
  \draw (0-\i,0-\i/\s) -- (\w*\n/2+\i/\s,0-\i/\s) -- (\w*\n/2+\i/\s,\h+\i/\s) -- (0+\i,\h+\i/\s);
  \draw (-0.5-\i,0-\i/\s) -- (-\w*\n/2-\i/\s,0-\i/\s) -- (-\w*\n/2-\i/\s,\h+\i/\s) -- (-0.5+\i,\h+\i/\s);
}
\end{tikzpicture}
\hskip 0.5cm
\begin{tikzpicture}[xscale=0.3,yscale=0.5]
\def\n{4}
\def\s{6}
\foreach \i in {0,...,\n}
{
	\draw (0+\i,0) -- (0.5+\i,0) -- (0.5+\i,4+\i/10) -- (4*\n,4+\i/10);
	\draw (0,-0.1-\i/10) -- (\n+\i+2,-0.1-\i/10) -- (\n+\i+2,1.9) -- (\n+\i+2.5,1.9);
	\draw (\n+1,2+\i/10) -- (2.5*\n+\i+1,2+\i/10) -- (2.5*\n+\i+1,3.9) --++ (0.5,0);
}
\end{tikzpicture}
\caption{\textbf{Left:} $K_{10}^-$ as 2-bend graph. \textbf{Right:} A clique is not contained in two edge-cliques.}
\label{fig-cliques}
\end{figure}
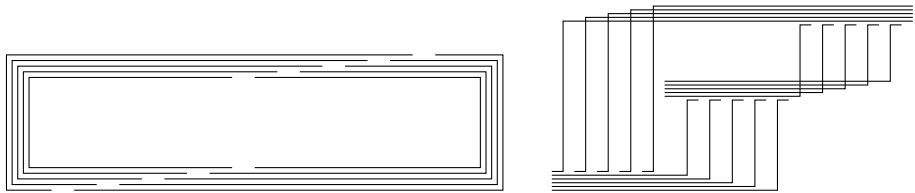

\section{Aligned 2-bend graphs}  \label{sec:aligned}
The main results of this section is the following complexity result.
\begin{theorem}
It is NP-complete to decide if a given graph is a 2-bend graph.
\label{thm:2linebend}
\end{theorem}

\begin{proof}
The NP-membership is obvious. As a polynomial certificate we use a list
of coordinates denoting start- and end-points of straight-line
segments. Such a representation has polynomial size with respect to the given graph.

For the NP-hardness we use a polynomial reduction from \pn3sat.
The instance of this problem is a set of clauses, each containing three non-negated variables. We ask for the existence of a truth assignment, such that each clause contains at least one true variable and at least one false variable (we say that such a clause is {\em satisfied}). The problem is NP-complete and equivalent to 2-coloring of 3-uniform hypergraphs (see Lov\'asz \cite{Lovasz}).

For a given formula $\varphi$, we shall construct a graph $G$, which is a 2-bend graph if and only if the formula is satisfiable. We start by replicating $\varphi$ 21 times (each time over a distinct copy of the set of variables), obtaining an equivalent formula $\varphi'$. The reason of this operation will be made clear in a while.

We start the construction of $G$ with two special vertices $a$ and $b$. Then for each variable $i$ of $\varphi'$, we add a vertex $v_i$ adjacent to both $a$ and $b$. For each occurence of $i$ in a clause $z$ of $\varphi'$, we add another vertex $o_{i,z}$, adjacent to $a$, $b$, and $v_i$.
Finally, for each clause $z=(i,j,k)$ we add four mutually non-adjacent vertices $c_z$, $d_z$, $e_z$, and $f_z$, with the following neighbors: $N(c_z) = \{o_{i,z}, o_{j,z}, o_{k,z}\}$; $N(d_z)=\{o_{i,z}, o_{j,z}\}$; $N(e_z)=\{o_{i,z}, o_{k,z}\}$; and $N(f_z)=\{o_{j,z}, o_{k,z}\}$ (see Fig.~\ref{aa:global} (left)).

\begin{figure}[ht]
\begin{tikzpicture}[xscale=2.5,yscale=1]
\node[circle, fill=black, fill=black,label={below:$d_z$}] (z12) at (0.5,1.5) {};
\node[circle, fill=black, fill=black,label={below:$f_z$}] (z23) at (1.5,1.5) {};
\node[circle, fill=black, fill=black,label={above:$c_z$}] (z123) at (1,1.9) {};
\node[circle, fill=black, fill=black,label={left:$e_z$}] (z13) at (1,2.7) {};
\node[circle, fill=black, fill=black,label={left:$o_{1,z}$}] (o1) at (0,1) {};
\node[circle, fill=black, fill=black,label={right:$o_{2,z}$}] (o2) at (1,1) {};
\node[circle, fill=black, fill=black,label={right:$o_{3,z}$}] (o3) at (2,1) {};
\node[circle, fill=black, fill=black,label={left:$v_1$}] (v1) at (0,0) {};
\node[circle, fill=black, fill=black,label={left:$v_2$}] (v2) at (1,0) {};
\node[circle, fill=black, fill=black,label={right:$v_3$}] (v3) at (2,0) {};
\node[circle, fill=black,label={left:$a$}] (a) at (0.5,-0.5) {};
\node[circle, fill=black,label={right:$b$}] (b) at (1.5,-0.5) {};
\draw (v1) -- (a);
\draw (v2) -- (a);
\draw (v3) -- (a);
\draw (v1) -- (b);
\draw (v2) -- (b);
\draw (v3) -- (b);
\draw (o1) -- (a);
\draw (o2) -- (a);
\path (o3) edge   [bend right] (a);
\path (o1) edge   [bend left] (b);
\draw (o2) -- (b);
\draw (o3) -- (b);
\draw (v1) -- (o1);
\draw (v2) -- (o2);
\draw (v3) -- (o3);
\draw (z12) -- (o1);
\draw (z12) -- (o2);
\draw (z23) -- (o2);
\draw (z23) -- (o3);
\path (z123) edge   [bend right] (o1);
\path (z123) edge   (o2);
\path (z123) edge   [bend left] (o3);
\path (z13) edge   [bend right] (o1);
\path (z13) edge   [bend left] (o3);
\end{tikzpicture}
\hfill\epsfbox{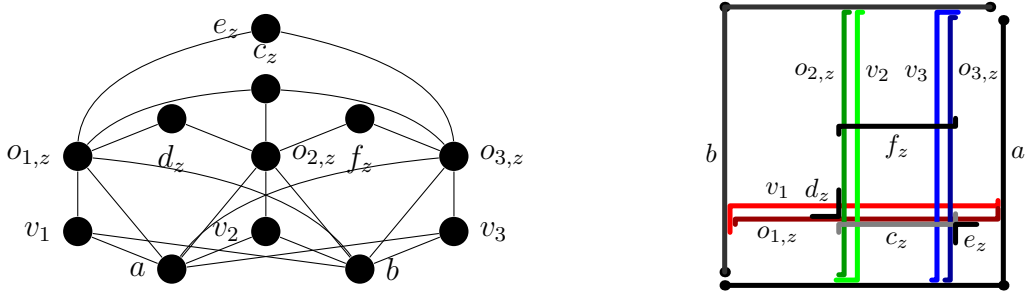}
\caption{\textbf{Left:} The graph obtained from a formula consisting of a single clause
$z=(1, 2, 3)$. For clarity we did not replicate the formula. 
\textbf{Right:} An EPG-representation of the graph on the left. The variable $1$ is false, while 2 and 3 are true.}
\label{aa:global}
\end{figure}

Now let us explain the main ideas behind the reduction. The purpose of vertices $a$ and $b$ is to cover the legs of each $P_{v_i}$ and $P_{o_{i,z}}$, keeping just their bodies exposed for possible intersections with clause-vertices.
This assumption may fail, as some $P_{v_i}$ or $P_{o_{i,z}}$ can be positioned over an end of a segment of $P_a$ or $P_b$, or on an intersection point of $P_a$ and $P_b$. However, each end can be occupied by paths corresponding to variables of only one copy of $\phi$. Moreover, each intersection point can interact with the representants of variables of at most two copies (see Fig. \ref{aa:6exits}).
\begin{figure}
\centering
 \epsfxsize=3.5cm\epsfbox{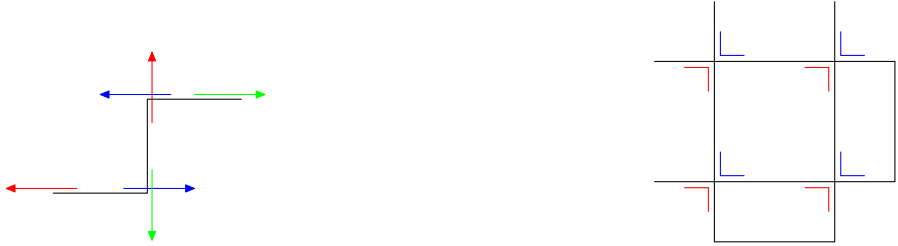}
\hskip 5cm
\begin{tikzpicture}[scale=0.8]
\draw (0,1) -- (4,1) -- (4,3) -- (0,3);
\draw (1,4) -- (1,0) -- (3,0) -- (3,4);
\draw[color=red] (0.5,0.9) --++ (0.4,0) --++ (0,-0.4);
\draw[color=blue] (1.1,1.5) --++ (0,-0.4) --++ (0.4,0);
\draw[color=red] (2.5,0.9) --++ (0.4,0) --++ (0,-0.4);
\draw[color=blue] (3.1,1.5) --++ (0,-0.4) --++ (0.4,0);
\draw[color=red] (0.5,2.9) --++ (0.4,0) --++ (0,-0.4);
\draw[color=blue] (1.1,3.5) --++ (0,-0.4) --++ (0.4,0);
\draw[color=red] (2.5,2.9) --++ (0.4,0) --++ (0,-0.4);
\draw[color=blue] (3.1,3.5) --++ (0,-0.4) --++ (0.4,0);
\end{tikzpicture}
\caption{\textbf{Left:} 6 pairwise non-adjacent segments may exit a 2-bend path without having to bend inside it. \textbf{Right:} At most 8 pairwise non-adjacent 2-bend paths may be adjacent to both $P_a$ and $P_b$ and contain their intersection point.}
\label{aa:6exits}
\end{figure}
As $P_a$ and $P_b$ have (together) 12 ends of segments and at most 4 intersection points, we have at most 20 special situations. But since  $\varphi$  is replicated 21 times, we are sure that for at least one copy of $\varphi$ our assumption holds (this type of trick we call the ``quantitative trick'' and we use it to cope with some obstructions which may appear only a constant number of times).
Moreover, note that if we want to make $P_{u_{i,z}}$ adjacent to both $P_a$ and $P_b$ using its body, it needs to cover either one of the ends, or on intersection point of $P_a$ and $P_b$, which can be excluded because of the ``quantitative trick''.
Therefore, if we want to keep the clause-vertices non-adjacent to $v_i$, one leg of each $P_{o_{i,z}}$ should be adjacent to $P_a$, the other one to $P_b$, and at least one of them also to $P_{v_i}$. This is summed up in the following claim.

\begin{clm} \label{clm-clean}
There is a copy of $\varphi$ in $\varphi'$, in which
\begin{enumerate}
\item each leg of each $P_{v_i}$ and $P_{o_{i,z}}$ are entirely covered by one of $P_a$ and $P_b$,
\item one leg of each $P_{o_{i,z}}$ is adjacent to $P_a$ and the other one is adjacent to $P_b$. Also, at least one of them has to be adjacent to $P_{v_i}$.
\end{enumerate}
\end{clm}

Let us focus on this ``clean'' copy of $\varphi$ in $\varphi'$. 
The body of each $P_{o_{i,z}}$ is exposed for representing clause-related vertices. Moreover, the orientation of the body (and thus of the whole path) is the same as the orientation of $P_{v_i}$. Thus we obtain the following.

\begin{clm} \label{clm-synchro}
For each variable $v_i$, all paths $P_{v_i}$ and $P_{o_{i,z}}$ (for every $z$) have the same orientiation, i.e., they are all either horizontal or vertical.
\end{clm}

The property above is crucial for the consistency of our construction, because the orientation of the paths will decide on truth assignment (horizontal means false, vertical means true). Whenever we say that two paths representing variables or their occurences are {\em synchronized}, we mean both equality of their truth-assignments and their horizontal/vertical layout.

First we show irrepresentabililty of the graph for an unsatisfiable formula.
Let $z=(i,j,k)$ be an unsatisfied clause. We will show that it cannot be represented.
Observe that it is impossible to have a 2-bend path adjacent to three parallel, pairwise non-collinear segments, while it is possible for two parallel and one perpendicular segments (see Fig.~\ref{aa:claumean}). 

\begin{figure}
\centering
\epsfxsize=5cm\epsfbox{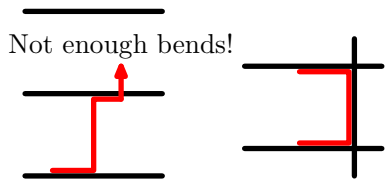}
\caption{It is impossible to intersect three parallel, pairwise non-collinear segments with a 2-bend path, while two parallel and one perpendicular segments can be intersected.}
\label{aa:claumean}
\end{figure}

The situation with three
parallel segments corresponds to all-true or all-false clause.
So, if no pair of middle segments of $P_{o_{i,z}}$, $P_{o_{j,z}}$, $P_{o_{k,z}}$ (and thus $P_{v_i}$, $P_{v_j}$, $P_{v_k}$) is collinear, we cannot represent $c_z$.

However, it might still happen that the bodies of, say, $P_{o_{i,z}}$ and $P_{o_{j,z}}$ are lying on the same line. But this pair of segments cannot be adjacent to more than one 2-bend path (see Fig.~\ref{aa:dvespojky} (left)). So if we represent $c_z$, then we cannot represent $d_z$, $e_z$, or $f_z$. This shows irrepresentability of an unsatisfied clause.

\begin{figure}
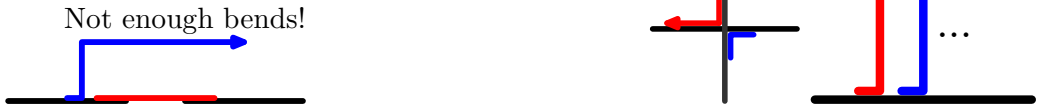

\epsfxsize=4cm\epsfbox{bendy.3}\hfill
\epsfxsize=2cm\epsfbox{bendy.4}
\epsfxsize=3cm\epsfbox{bendy.5}
\caption{{\bf Left:} Two collinear segments cannot be adjacent to two
mutually non-adjacent 2-bend paths.
{\bf Right:} This is possible for two mutually intersecting segments or two non-collinear parallel segments.}
\label{aa:dvespojky}
\end{figure}

For a representable formula, we build a canonical representation shown in Fig.~\ref{aa:global} (right) (for a clause with one false and two true literals, we rotate everything except for $a$ and $b$ by 90 degrees). Figure~\ref{aa:global} shows one clause in one of 21 replicated copies and one occurence of each variable. The full construction with all 21 copies would consist of 21 copies of all items present in the picture, except for $P_a$ and $P_b$. Note that there are no edges between vertices belonging to different copies of $\varphi$. Further occurences, e.g., of $v_2$ in the same formula can be represented next to $o_{2,z}$ intersecting $v_2$ in the bottom (or top) horizontal leg (where it simultaneously intersects $a$ or $b$, respectively, and it has to avoid legs of other possible occurences). Anyway, their truth assignments are synchronized in all possible cases as they have to intersect $a$ or $b$ together with the vertex representative $v_2$. Considering two (and more) clauses in the representation, each clause has its own occurences, so the representation of one clause does not influence representations of other clauses (as representatives of distinct occurences are not mutually adjacent, i.e., they are disjoint up to finitely many points).
 In this representation, the body of each $P_{o_{i,z}}$ intersects the body of each $P_{o_{j,z}}$, for all $i$ evaluated to true and $j$ evaluated to false. Thus it is possible to represent all clause-vertices, just as depicted.
\end{proof}

\subsection{Subclasses of aligned 2-bend graphs} \label{ss:subclasses}

Here we focus on the recognition of particular subclasses of 2-bend graphs. Note that as there are many classes (whose recognition
is often NP-hard), it is important to ask whether even some polynomially
recognizable class can exist ``in between''. This concept is called {\em sandwiching}.
Formally, having two classes of graphs ${\cal A}\subseteq {\cal B}$, a class
$\cal C$ is {\em sandwiched between $\cal A$ and $\cal B$} if ${\cal A}
\subseteq {\cal C} \subseteq {\cal B}.$
For optimization problems, it holds that if an algorithm works for
class $\cal B$, it works also for the class $\cal A$. Also a hardness result for $\cal A$ carries over to $\cal B$. However, the recognition
problem behaves in a different way. As a trivial example we may pick a class $\cal A$
containing only complete graphs (this class is polynomially recognizable),
for class $\cal B$ we may take class of all graphs (which is also polynomially
recognizable) and between them we can find, e.g., classes of 2-bend graphs,
whose recognition is NP-complete, as shown in Theorem \ref{thm:2linebend}. Similarly,
between two NP-hard classes, a polynomially-recognizable class can be
sandwiched (consider e.g. 3-colorable planar graphs, planar graphs, and 4-colorable graphs).

In this section we do not only show the recognition hardness of individual
classes, but we are trying to find a smallest class $\cal A$ and a largest class $\cal B$, such that no polynomially-recognizable class can be sandwiched between them.

We start with first two subclasses where our reduction
for 2-bend graphs can be applied directly. One of them
is a class of monotonic 2-bend graphs (i.e., $\{\RZ,\N\}$-graphs) and the other is the class of $\{\RZ,\RN\}$-graphs. 

We observe that in the proof of  Theorem~\ref{thm:2linebend} we produce
a monotonic 2-bend graph from each satisfiable formula. As a
non-satisfiable formula cannot be represented by any 2-bend graph,
if there was a polynomially-recognizable class between monotonic 2-bend
graphs and 2-bend graphs, we would be able to distinguish satisfiable
formulae from non-satisfiable ones, showing P=NP.

It is very simple to redraw the representation used in the proof of  Theorem~\ref{thm:2linebend}, using only \RZ\; and \RN\;-shapes.

\begin{corollary} \label{cor:mono-zig-zag}
It is NP-complete to recognize monotonic 2-bend graphs and $\{\RZ,\RN\}$-graphs.
Moreover, between 2-bend graphs and any of these classes, or even their intersection, 
no polynomially recognizable class can be sandwiched (unless P=NP).
\end{corollary}

Now we shall modify the construction a bit to show a cascade of further
results. 
Note that there are four possible patterns of horizontal paths (\RU, \U, \N, \RN) and another four for vertical paths. As we want to show that it is NP-complete to recognize graphs of any class $X \in \{\RU, \U, \N, \RN\} \times \{\C,\RC,\Z,\RZ\}$, we need to start with exploring the symmetries, to classify possible classes $X$.

So consider a pair or shapes, one of which is horizontal and the other one is vertical.
If both legs of each shape bend in the same direction, we obtain the class $\{\RU,\RC\}$, which is equivalent to each $\{\U,\RC\}$, $\{\RU,\C\}$, and $\{\U,\C\}$ (consider a rotation of flipping of an EPG-representation).
If both legs of one shape bend in the same direction, and the legs of the other shape bend in opposite directions, we get the class $\{\RU, \Z\}$ (again, up to symmetry). Finally, if the legs of both shapes bend in opposite directions, we get two possibilities, i.e., $\{\RZ, \N\}$ (monotonic 2-bend graphs) and $\{\RZ, \RN\}$. Although for the latter two classes we have already shown NP-hardness, now we show yet one construction that works for all four cases. Such a general construction is important from the point of view of sandwiching.

The new construction, in fact, is just a simplified version of the one in the proof of Theorem \ref{thm:2linebend}. Again, for a formula $\varphi$, we replicate it to obtain $\varphi'$ (using ``quantitative trick'', see Claim \ref{clm-clean}) and introduce variable-vertices $v_i$ and occurence-vertices $o_{i,z}$. The difference is that now each clause $z=(i,j,k)$ is represented by just one vertex $c_z$, adjacent to $o_{i,z}, o_{j,z}$, and $o_{k,z}$ (so we omit vertices $d_z, e_z$, and $ f_z$). For a formula $\varphi$, let us call such constructed graph $G(\varphi')$.

Using this construction we can show that it is NP-complete to recognize $X$-graphs for each of the pairs $X$ of permitted shapes, one of which is vertical and the other horizontal. Let us start with proving the following statement.

\begin{lemma} \label{lem:rep}
If $\varphi$ is a satisfiable \pn3sat formula, then $G(\varphi')$ can be represented by any of the following pairs of shapes: $\{\RU,\RC\}, \{\RU,\Z\}, \{\RZ, \N\}, \{\RZ, \RN\}$.
\end{lemma}

\begin{proof}
We will represent $G(\varphi')$ in a way similar to Fig. \ref{aa:global} (right), i.e., truth assignment of the variable $i$ reflects whether we use the vertical or the horizontal shape to represent $v_i$ and all $o_{i,z}$'s. To represent vertices $a$ and $b$, we need two opposite right angles. This way we create a ``frame'' from the canonical representation.
Then, the vertical (horizontal) shape can represent variables and occurences assigned true (false, resp.). The synchronization of occurences and variables works in exactly the same way as described in the proof of Theorem \ref{thm:2linebend}.
Consider a clause $z=(i,j,k)$. Without loss of generality assume that $i$ and $j$ are assigned true and $k$ is assigned false (all other cases are symmetric). Then $o_{i,z}$ and $o_{j,z}$ are represented by vertical shapes (with non-collinear bodies), while $o_{k,z}$ is represented by the horizontal shape intersecting the bodies of both $P_{o_{i,z}}$ and $P_{o_{j,z}}$. Since the horizontal shape can be made adjacent to each $P_{o_{i,z}}$, $P_{o_{j,z}}$, and $P_{o_{j,k}}$, we can represent each clause vertex $c_z$.
\end{proof}

Now we are ready to show the following.

\begin{lemma}\label{lem:pairs}
It is NP-complete to recognize $X$-graphs, for any $X \in \{\RU, \U, \N, \RN\} \times \{\C,\RC,\Z,\RZ\}$.
\end{lemma}

\begin{proof}
By Lemma \ref{lem:rep}, we know that if $\varphi$ is a satisfiable formula, then $G(\varphi)$ is an $X$-graph for any $X$. To complete the proof we will show that if $G(\varphi')$ is an $X$-graph, then $\varphi$ is a satisfiable \pn3sat formula.

By the ``quantitative trick'' we can assume that in this representation no pathological situations happen, i.e., Claim \ref{clm-clean} holds. Also, by Claim \ref{clm-synchro}, we know that each $P_{o_{i,z}}$ is synchronized with $P_{v_i}$ (they are all either horizontal or vertical).

We consider possible (up to symmetry) classes $X$ separately.

\noindent \textbf{Case 1. $X = \{\RU, \RC\}$.}

First, suppose that no pair of bodies of occurence-representatives is collinear. If for every clause $z=(i,j,k)$, at least one of $P_{o_{i,z}},P_{o_{j,z}},P_{o_{k,z}}$ is represented by a $\RU$-shape and at least one by a $\RC$-shape, then we can find an assignment satisfying $\varphi$, which contradicts our assumption. On the other hand, if the bodies of $P_{o_{i,z}},P_{o_{j,z}},P_{o_{k,z}}$ are parallel and pairwise non-collinear, then by the previous reasoning, we cannot construct $P_{c_z}$.

So suppose that the bodies of $P_{o_{i,z}}$ and $P_{o_{i',z'}}$ are one the same line $\ell$ (without loss of generality suppose that $\ell$ is horizontal and $P_{o_{i,z}}$ is left of $P_{o_{i',z'}}$).

We observe that both $P_a$ and $P_b$ are $\RU$-shapes, and $\ell$ is intersected by  $P_a$ and $P_b$ (at least) four times,  in the ordering $a, b, a, b$ or $a, b, b, a$ (or symmetric, but the pattern $a,a,b,b,$ is forbidden, as $o_{i,z}$ and $o_{j,z}$ are non-adjacent). This follows from the fact that, by Claim \ref{clm-clean}, one leg of $P_{o_{i,z}}$ ($P_{o_{i',z'}}$, resp.) shares a grid edge with some segment of each $P_a$ and $P_b$, and wastes one bend point for each.

The only (topological) possibilities of representing $P_a$ and $P_b$ with $\RU$-shapes are depicted in Fig.~\ref{fig:abab-abba}. Since false (true) variables and their occurences are represented by \RC-shapes (\RU-shapes, resp.), the appropriate paths have to pass between $P_a$ and $P_b$. The red zones in Fig. \ref{fig:abab-abba} depict the only place where representatives of true variables can be placed, while blue zones show where false variable representatives can occur. Note that the blue zone can be extended downwards, but not upwards. The blue and red zones are disjoint (up to one segment that cannot be used by clause representatives). Thus horizontal and vertical representatives of variable-occurences cannot intersect.
\begin{figure}
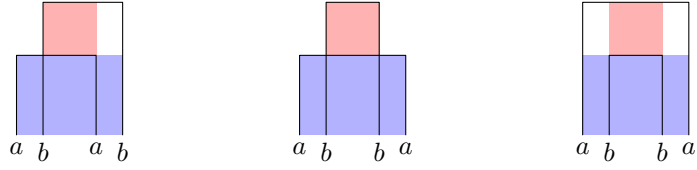

~\hfill\epsfbox{bendy.11}\hfill\epsfbox{bendy.12}\hfill\epsfbox{bendy.13}\hfill~
\caption{The possible positions of $P_a$ and $P_b$.}
\label{fig:abab-abba}
\end{figure}
Consider a clause $z'' = (i'',j'',k'')$, where at least one of the variables is represented by a $\RU$-shape and at least one is represented by a $\RC$-shape and no pair of occurence representatives intersect each other. It is easy to observe that we cannot represent $c_{z''}$ by a \RU-shape or by a \RC-shape, so that it is adjacent to all three $P_{o_{i'',z''}},P_{o_{j'',z''}},P_{o_{j'',z''}}$ in the desired way, i.e., avoiding creating adjacency to $a$ and $b$.

So the representatives or variables (and occurences) are either all horizontal or all vertical. If they are vertical, then they all must be parallel (see Fig. \ref{fig:abab-abba}), and again we cannot represent any clause vertex.

So, finally, we assume that all variable (and occurence) vertices are represented by $\RU$-shapes. Suppose there are two clauses $z = (i,j,k)$ and $z'=(i,j',k')$, such that  (i) the bodies of $P_{o_{i,z}}$ and $P_{o_{j,z}}$ are collinear and the body of $P_{o_{i,z}}$ is left of  the body of $P_{o_{j,z}}$, (ii) the bodies of $P_{o_{i,z'}}$ and $P_{o_{j',z'}}$ are collinear and the body of $P_{o_{i',z'}}$ is right of  the body of $P_{o_{j',z'}}$. Using the limited possibilities of representing variable vertices (see Fig. \ref{fig:abab-abba}), by the case analysis we observe that such a situation is impossible. 

Consider a clause $z = (i,j,k)$. Observe that always two of $P_{o_{i,z}}$ and $P_{o_{j,z}}$ have to be collinear (to make it possible to intersect them together with yet one horizontal curve by a vertical curve).
Thus we can define the following truth assignment. For each clause $z=(i,j,k)$, if the bodies of $P_{o_{i,z}}$ and $P_{o_{j,z}}$ are collinear and the body of $P_{o_{i,z}}$ is left to the body of $P_{o_{j,z}}$, then $i$ is assigned true and $j$ is assigned false (note that no variable is set both true and false). The unassigned variables can get arbitrary values. Since each clause contains at least one true and least one false variable, our truth assignment satisfies $\varphi$, which contradicts our assumption.

\noindent \textbf{Case 2. $X = \{\RU, \Z\}$.}

Suppose we have an unsatisfiable formula $\varphi$ and a  $\{\RU, \Z\}$-representation of $G(\varphi')$. If no pair of occurences has collinear bodies, we have shown that $G(\varphi')$ cannot be represented. Also note that no pair of vertical segments representing variables and occurences can be represented on the same vertical line, since then it is impossible to construct $P_a$ and $P_b$.

Thus suppose we have two occurence-vertices, represented by $\RU$-shapes, whose bodies are collinear. But then neither $a$ nor $b$ can be represented by a $\Z$-shape, so both of them have to be represented by a $\RU$-shape and we perform exactly the same as  previously (because of exactly the same argument, we can use no $\Z$-shape to represent a variable-vertex or an occurence-vertex, and the truth assignment of $\RU$-shapes is exactly the same as in the previous case).

\noindent \textbf{Case 3. $X=\{\RZ, \N\}$ or $X = \{\RZ, \RN\}$}

Each curve permitted in this case is monotonic with respect to some pair of directions. Thus none of them crosses any (vertical or horizontal) line more than once. So the bodies of no variable or occurence representants can be collinear. Thus we cannot represent an unsatisfied clause. \end{proof}

Note that the lemma above shows that, both, an intersection and a union of the mentioned subclasses (as well as anything sandwiched between them) is NP-hard to get recognized. Also, note that it does not show that all classes representable by a given subset of 2-bend shapes (which includes at least one vertical and at least one horizontal shape) are NP-complete to get recognized. It still may happen that there exists such a set $X$ of patterns, that $X$-graphs can be polynomially recognized. However, we know that if such a class exists, it must not contain even the intersection of  $\{\RZ, \N\}$-graphs and $\{\RZ, \RN\}$-graphs.

Finally, let us try to explore the limits of the original hardness reduction for 2-bend graphs (Theorem \ref{thm:2linebend}). We know that it works for 2-bend graphs, for $\{\RZ, \N\}$-graphs, and for $\{\RZ, \RN\}$-graphs (and where the inclusion-relation applies, then also for everything in between). However, we may show that the reduction works also for all triples of 2-bend shapes, in which at least one shape is vertical, at least one is horizontal, and they are not symmetric to the triple $\{\RC, \Z, \U\}$, i.e., without loss of generality, two vertical shapes, one having its legs in the same direction, the other having legs in mutually opposite directions, and the legs of the horizontal one go in the same direction and yet in the direction ``towards the common angle'' of the other two gadgets. It is easy to observe that the ``simplified'' construction can be represented, so we need to show, for a particular satisfied clause $z = (i,j,k)$, how to represent vertices $d_z, e_z$, and $f_z$. Suppose without loss of generality $i,j$ are evaluated true and $k$ is evaluated false. The path $P_{c_z}$ passes through the intersection point of $P_{o_{i,z}}$ and $P_{o_{k,z}}$, and through the intersection point of $P_{o_{j,z}}$ and $P_{o_{k,z}}$. In order to represent $d_z$ (adjacent to $o_{i,z}$ and $o_{j,z}$) we need to use the same intersection-point, i.e., we need the angle obtained from $c_z$ rotated by 180 degrees.

\begin{lemma}
It is NP-complete to recognize $X$-graphs, where $X$ is
any triple of 2-bend shapes containing at least one vertical and one horizontal shape, and is not symmetric to $\{\RC, \Z, \U\}$.
\end{lemma}

\begin{proof}
First let us characterize these triples $X$. To filter out symmetries, we start analyzing the shapes. We know that, without loss of generality, two shapes are vertical and one horizontal. There are three possibilities for the choice of vertical shapes: either each vertical shape has both its legs in the same direction, or one has legs in the same direction while the other not, or both of them have legs in different directions. This gives us 3 possibilities (up to symmetries) of vertical pairs: $\{\RC, \C\}$, $\{\RC, \Z\}$, and $\{\Z, \RZ\}$ (note that $\{\RC, \RZ\}$ is symmetric to $\{\RC, \Z\}$).
For each of them we try all four possible horizontal segments and, again, we obtain some symmetries: $\{\RC, \C, \U\}$ is symmetric to $\{\RC,\C, \RU\}$, $\{\RC, \C, \N\}$ is symmetric to $\{\RC, \C, \RN\}$, and $\{\Z, \RZ, \U\}$ is symmetric to $\{\Z, \RZ, \RU\}$ (always vertical flip).

One of the cases is the excluded one ($\{\RC, \Z, \U\}$), so we have to analyze eight cases. First four of them are simple: $\{\RC, \Z, \N\}$, $\{\RC, \Z, \RN\}$, $\{\Z, \RZ, \N\}$, and $\{\Z, \RZ, \RN\}$ contain either $\{\RZ,\RN\}$-graphs or $\{\RZ,\RZ\}$-graphs, so the reduction works.

The remaining cases are $\{\RC, \C,\U\}$, $\{\RC, \C, \N\}$, $\{\RC, \Z, \RU\}$, and $\{\Z, \RZ, \U\}$, where we have to employ a kind of brute force.

Consider a clause $z = (i,j,k)$. We need to show how to represent the vertices $c_z, d_z, e_z$, and $f_i$, separately for the case when two variables from $z$ are true (and one false), and when two variables from $z$ are false (and one true).

In the canonical representation, occurence-representatives of true variables intersect all occurence-representatives of false variables (in just one point). To represent (without loss of generality) $c_z, d_z$, and $e_z$, we have to use these these intersection points simultaneously by $c_z, d_z$ and $c_z, e_z$ (one with the former pair, one with the latter). So for $d_z$ and $e_z$ we need a shape with an angle ``opposite'' to the shape used for representing $c_z$. So it suffices to show that there exist one vertical and one horizontal segment, such that both its angles have their ``opponents''.

\noindent  {\bf Case 1: $X=\{\RC, \C,\U\}$.}

Suppose first that two variables from $z$ (say $i$ and $k$) are true and one ($j$) is false. Thus the bodies of $P_{o_{i,z}}$ and $P_{o_{k,z}}$ will be vertical, while the body of $P_{o_{j,z}}$ will be horizontal, which implies that $c_z$ can be represented with a \U-shape. Its left angle is opposite to the upper angle of a \RC-shape, its right angle is opposite to the upper angle of a \C-shape. Therefore these shapes can be used for representing (without loss of generality) $d_z$ and $e_z$. A \U-shape can be also used for representing $f_z$.

If two variables are false, we may use a \RC-shape, and as the opposite angles we use the angles of a \C-shape. This completes the analysis for this triple.

\noindent  {\bf Case 2: $X=\{\RC, \C,\N\}$.}

The clause with two true variables is the same as in the previous case (we use a \RC-shape and a \C-shape). The clause with two false variables has to be represented by a \N-shape. Its left angle is opposite to the lower angle of a \RC-shape, and the right angle is opposite to the upper angle of a \C-shape.

\noindent  {\bf Case 3: $X=\{\RC, \Z, \RU\}$.}

Lower angles of a \RC-shape and a \Z-shape are the opposite angles for a \RU-shape. To represent vertical segments we pick a \Z-shape. Note that its angles are mutually opposite (one to another).

\noindent  {\bf Case 4: $X=\{\Z, \RZ, \U\}$.}

This triple is also simple: \Z-shapes and \RZ-shapes give us all possible angles.

\noindent  {\bf Case 5: $X=\{\RC, \Z, \U\}$ (excluded triple).}

Finally, let us remark that in this triple, we cannot find an angle opposite to one angle of a \U-shape (because both upper angles of a \RC-shape and a \Z-shape are going in the same direction). Thus this triple is exceptional and our reduction does not work for it.
\end{proof}

As a corollary of the last lemma, the reduction works for all such 4-tuples of 2-bend shapes, where at least one shape is vertical and at least one horizontal (a non-trivial situation arises only when extending $\{\RC, \Z, \U\}$). Note also that the reduction works for any $k$-tuple of 2-bend shapes for $k \geq 5$ (as there are just 4 vertical and 4 horizontal shapes, we are sure that at least one will be horizontal and at least one will be vertical).

Summing up the results from this section, we obtain the following.
\begin{theorem} \label{thm:classes and subclasses}
It is NP-complete to recognize $X$-graphs, where $X$ is:
\begin{enumerate}[(i)]
\item any of $\{\RU, \U, \N, \RN\} \times \{\C,\RC,\Z,\RZ\}$,
\item any triple of 2-bend shapes containing at least one vertical and one horizontal shape, and is not symmetric to $\{\RC, \Z, \U\}$,
\item any 4-tuple of 2-bend shapes, containing at least one horizontal and one vertical shape,
\item any $k$-tuple of 2-bend shapes for $k \geq 5$.
\end{enumerate}

Moreover, it is impossible to sandwich any polynomially recognizable class between:
\begin{enumerate}[(a)]
\item the intersection of $\{\RU, \U, \N, \RN\} \times \{\C,\RC,\Z,\RZ\}$ and their union,
\item intersection of classes given in (ii) and the class of 2-bend graphs.
\end{enumerate}
\end{theorem}

\section{More slopes}  \label{sec:unaligned}
In this section we relax the definition of an EPG-representation. By an {\em unaligned EPG-representation} of a graph $G$ we mean a mapping from vertices of $G$ to a set of polylines (piecewise linear curves), such that two vertices are adjacent if and only if their corresponding polylines share infinitely many points. Again, we are interested in keeping the number of bends (or equivalently, segments in a polyline) small.

Here we show hardness of the recognition of unaligned 2-bend graphs and conclude the section with discussion of a trade-off between the number of slopes used and the number  of bends.
\begin{theorem}
It is NP-hard to recognize unaligned 2-bend graphs.
\label{thm:unalignedrecog}
\end{theorem}

\begin{proof}
This time we reduce from {\sc 3-Coloring}.
For a graph $G$ we shall construct a graph $H$, which is an unaligned 2-bend graph if and only if $G$ is 3-colorable.

The reduction uses ideas similar to the reduction for aligned 2-bend
graphs. This time we use 12 service vertices and again we want our gadgets
to avoid being represented over the ends of segments of these service
vertices, and over their mutual intersection points. So we use the ``quantitative trick'' again. This time we may have no more than 1 260
special places ($12 \cdot 2 \cdot 3$ ends of segments, $\binom{12}{2} \cdot 9$ possible intersection points, each of which can be used at most twice). Thus we take 1 261 disjoint copies of the graph $G$, obtaining the graph $G'$, which is 3-colorable if and only if $G$ is 3-colorable.

The main idea of the reduction is that one service vertex of $H$, named $a$,
simulates the 3-coloring of $G'$. The individual segments of $P_a$
correspond to three color classes.
Each vertex $v$ of $G'$ will be represented by several vertices of $H$. One of them, called $v_2$, will have the property that one of the legs of $P_{v_2}$ lies on a segment of $P_a$ (thus defining the color of $v$ in a 3-coloring of $G'$), and the remaining two segments of $P_{v_2}$ will be fully covered by some other paths, non-adjacent to edge-representatives.
An edge $uv$ of $G'$ will be represented by a pair of mutually non-adjacent vertices of $H$. Both of them will be made adjacent to $a$ and the representatives of both $u$ and $v$. The main idea is that we cannot construct edge-representatives, if $v_2$ and $u_2$ are adjacent to the same segment of $a$ (and thus $v$ and $u$ get the same color). 
This part of $H$ is illustrated in Fig.~\ref{gd:idea}.

\begin{figure}
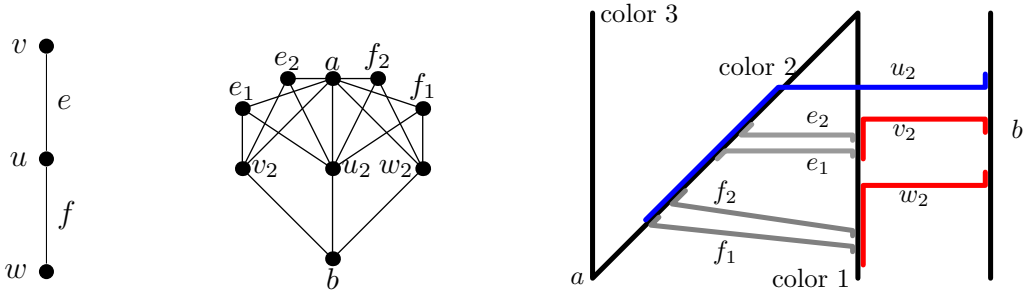

\begin{tikzpicture}[yscale = 1.5]
\node[circle, fill=black,label={left:$w$},inner sep=0pt,minimum size=0.2cm] (w) at (0,0) {};
\node[circle, fill=black,label={left:$u$},inner sep=0pt,minimum size=0.2cm] (u) at (0,1) {};
\node[circle, fill=black,label={left:$v$},inner sep=0pt,minimum size=0.2cm] (v) at (0,2) {};
\draw (w) -- (u) node[midway, right] {$f$};
\draw (u) -- (v) node[midway, right] {$e$};
\end{tikzpicture}
\hfill
\epsfxsize=2.8cm \epsfbox{bendy.16}\hfill\epsfbox{bendy.7}
\caption{\textbf{Left:} The graph $G$. \textbf{Middle:} The main part of $H$. For clarity, just the main vertex-representants are depicted. Also the replication (``quantitative trick'') was not performed.
\textbf{Right:} An unaligned 2-bend representation of $H$. Note that having fixed representations of $v_2$ and $u_2$,  we are unable to represent the edge $vw$ in such a way that it representatives are adjacent to $v_2$ and $u_2$ only on their legs lying on $a$.}
\label{gd:idea}
\end{figure}

Formally, the graph $H$ has 12 service vertices $a_0, a_{0.5}, a_1, a_{1.5}, a_2, a_{2.5}, a_3, a_{3.5}$, $a, b, a_B$, and $ b_B$.
For each vertex $v$ of $G'$, we add to $H$ vertices $v_1, v_{1.5}, v_2, v_{2.5}, v_3$, and $v_b$ (we will call them $v$-vertices).
The vertex $v_b$ is adjacent to all other $v$-vertices. Furthermore, $v_{1.5}$
is adjacent to $v_1, v_2$, and $v_{2.5}$ is adjacent to $v_2, v_3$.
Finally, each $v$-vertex is adjacent to two service vertices:
$v_1$ to $a_0,a_1$, $v_{1.5}$ to $a_{0.5},a_{1.5}$,
$v_2$ to $a,b$, $v_{2.5}$ to $a_{2.5},a_{3.5}$, $v_3$ to $a_2,a_3$.
For each edge $e=uv$ we add a pair of mutually non-adjacent vertices $e_1, e_2$, both
adjacent to $a$, $u_2$, and $v_2$.

Suppose we have an unaligned 2-bend representation of $H$. First, by the ``quantitative trick'', we know that at least for one copy of $G$, for any vertex $v$, all vertices $v_i$ ($i\in\{1, 1.5, 2, 2.5, 3, b\}$) are represented by 2-bend paths having both legs covered by the segments of the appropriate pair of service vertices. Let us focus on this copy of $G$.

We observe that the body of $P_{v_2}$ (for any $v$)
is covered by (at least) $P_{v_b}$. This follows from
the fact that $P_{v_b}$ can intersect the other $v$-vertices
only by its body (as one leg lies on $P_{a_B}$, and the second on $P_{b_B}$).
Thus the bodies of $P_{v_1}, P_{v_{1.5}}, \ldots, P_{v_3}, P_{v_b}$
 must form an interval representation of $H[\{v_1, v_{1.5}, \ldots, v_3, v_b\}]$ and in no such representation the body of $P_{v_2}$
can exceed the body of $P_{v_b}$. Therefore the body
of $P_{v_2}$ is fully covered by (at least) the body of
$P_{v_b}$.

Now, we are in a desired situation. Consider an edge $e=uv$. For each $P_{u_2}$ and $P_{v_2}$, only the leg lying on $P_a$, can be made adjacent to both $P_{e_1}$ and $P_{e_2}$, as using any other segment would cause some unwanted adjacency.
If these legs are on distinct segments of $P_a$, obviously
we can represent both $e_1$ and $e_2$. Conversely, if they are on
the same segment of $P_a$, we can represent at most one of them (similarly to Fig.~\ref{aa:dvespojky} (left)). This shows irrepresentability for a non-3-colorable $G$.

On the other hand, if $G$ has a 3-coloring, we use it for distributing
segments of $P_{v_2}$ of each vertex $v$ over the segments of $P_a$.
Note that we may create a representation,
where the bodies of $P_{v_2}$, for all $v$, are parallel.
Then other $v$-vertices may be represented in the way shown in Fig.~\ref{gd:global}. For any edge $e$, paths $P_{e_1}$ and $P_{e_2}$ connect two non-collinear segments, which can be easily done. \end{proof}

\begin{figure}
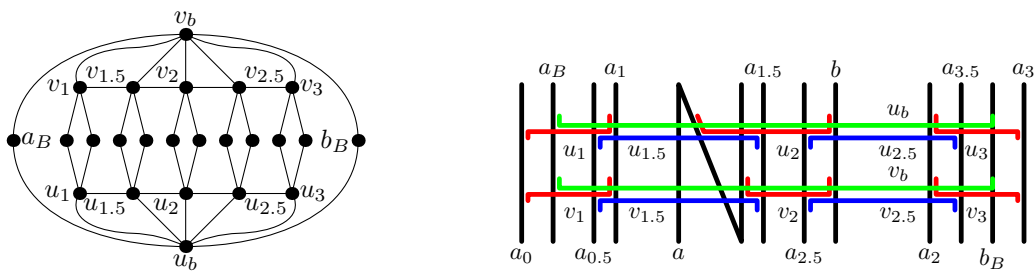

\epsfbox{bendy.17}\hfill
\epsfxsize=7cm \epsfbox{bendy.8}
\caption{
\textbf{Left:} A graph $H$ for $G$ being an edge $uv$ (replication is omitted).
Unlabeled vertices between $a_B$ and $b_B$ are, respectively: $a_0, a_1,
a_{0.5}, a_{1.5}, a, b, a_{2.5}, a_{3.5}, a_2, a_3$.
\textbf{Right:}~Unaligned 2-bend representation of $H$.
}
\label{gd:global}
\end{figure}

\subsection{Slopes and bends}
Defining unaligned bend graphs permits us to introduce a new
measure of complexity of a representation, namely, the number of slopes used.
There is an obvious trade-off between the number of bends and the number of slopes.
Before we explore this relation a little more, let us try to minimize the number of different slopes used by the unaligned 2-bend representation.

\begin{proposition} \label{lower-slopes}
In order to represent all unaligned 2-bend graphs on $n$ vertices,
we need $\Omega(\sqrt{n})$ slopes.
\end{proposition}
\begin{proof}
The proof follows from the construction in the proof of Theorem \ref{thm:unalignedrecog}. Let $G \sim K_{m,m,1}$ be a complete bipartite graph
with biparition classes $X,Y$, both of size $m$, and one extra vertex $z$ adjacent to all other vertices.

Th graph $G$ is replicated 1261 times, obtaining $G'$, and construct $H$ in the way described in the proof of Theorem \ref{thm:unalignedrecog}. Since $G$ has $2m+1$ vertices and $\Theta(m^2)$ edges, $H$ has $n=\Theta(m^2)$ vertices.

As $G$ is 3-colorable, $H$ has an unaligned 2-bend representation. 
As always, we will focus on the ``clean'' copy of $G$. 
Consider the path $P_a$, and let $p,q,r$ denote its three segments.
By the properties of $H$, without loss of generality one leg of every $P_{x_2}$ for $x \in X$ lies on $p$, while one leg of every ${P_{y_2}}$ for $y \in Y$ lies on $r$.

Now consider the paths $P_{e_1}$ (for $e = xy$, $x \in X$, $y \in Y$).
There are $m^2$ such paths. We observe that every slope $\ell$ can be used by the bodies of at most $2m$ paths $P_{e_1}$. To see this, we use a sweeping line, parallel to $\ell$. As each path $P_{e_1}$ connects a pair of segments of a different pair $(P_{x_2},P_{y_2})$, the sweeping line must leave at least one of the segments before meeting a new one. As there are in total $2m$ segments of $P_{x_2}$ or $P_{y_2}$ on $P_a$, at most $2m$ paths $P_{e_1}$ can have their bodies parallel to $\ell$.
Thus we need at least $\left \lceil \frac{m^2}{2m} \right \rceil = \Theta(m) = \Theta(\sqrt{n})$ different slopes to represent the bodies of paths $P_{e_1}$. 
\end{proof}

To see a trade-off between the number of bends and the number of slopes, observe that for $G \sim K_{m,m,1}$, the graph $H$ can be easily represented by 3-bend paths, using only two slopes ($P_a$ is represented by a \O-shape with segments of $P_{v_2}$ on three different segments of it).
Now we can represent $P_{e_1},P_{e,2}$ (see Fig.~\ref{fig:3bend}), 
and finish the construction as in Fig.~\ref{gd:global}.

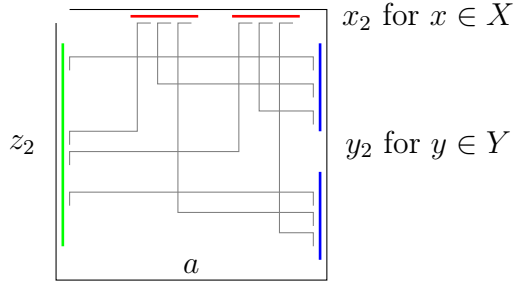
\begin{figure}
\centering
\begin{tikzpicture}[scale=0.9]
\path[draw] (0,3.8) -- (0,0) -- (4,0) -- (4,4) -- (0.2,4);
\path[draw, color=red, line width=1] (1.1,3.9) --++ (1,0);
\path[draw, color=red, line width=1] (2.6,3.9) --++ (1,0);
\path[draw, color=blue, line width=1] (3.9,3.5) --++ (0,-1.3);
\path[draw, color=blue, line width=1] (3.9,1.6) --++ (0,-1.3);
\path[draw, color=green, line width=1] (0.1,3.5) --++ (0,-3);
\path[draw, color = gray] (0.2,2) --++ (0,0.2) --++ (1,0) --++ (0,1.6) --++ (0.2,0);
\path[draw, color = gray] (0.2,1.7) --++ (0,0.2) --++ (2.5,0) --++ (0,1.9) --++ (0.2,0);
\path[draw, color = gray] (0.2,3.1) --++ (0,0.2) --++ (3.6,0) --++ (0,-0.2);
\path[draw, color = gray] (0.2,1.1) --++ (0,0.2) --++ (3.6,0) --++ (0,-0.2);
\path[draw, color = gray] (1.7,3.8) --++ (-0.2,0) --++ (0,-0.9) --++ (2.3,0) --++ (0,-0.2);
\path[draw, color = gray] (2,3.8) --++ (-0.2,0) --++ (0,-2.8) --++ (2,0) --++ (0,-0.2);
\path[draw, color = gray] (3.2,3.8) --++ (-0.2,0) --++ (0,-1.3) --++ (0.8,0) --++ (0,-0.2);
\path[draw, color = gray] (3.5,3.8) --++ (-0.2,0) --++ (0,-3.1) --++ (0.5,0) --++ (0,-0.2);
\node at (2,0.2) {$a$};
\node at (-0.5,2) {$z_2$};
\node at (5.5,2) {$y_2$ for $y \in Y$};
\node at (5.5,3.9) {$x_2$ for $x \in X$};
\end{tikzpicture}
\caption{The idea of a 3-bend representation of $P_a$ and $P_{e_1},P_{e_2}$.
}
\label{fig:3bend}
\end{figure}

\subsection{$d$-bend number}
Let us conclude the section with some generalization of the bend number. Fix a set $D$ of $d$ pairwise non-parallel lines (slopes). We say that an unaligned EPG-representation is an {\em $EPG(D)$-representation} if every segment of each polyline is parallel to some line in $D$.

The {\em $d$-bend number} $b_d(G)$ of a graph $G$ is the minimum $k$ for which there exists a set $D$ of $d$ slopes, such that $G$ has an $EPG(D)$-representation  in which every path bends at most $k$ times. We also define $b_{\infty}(G): = \min \limits _{d \in \NN} b_d(G)$, which corresponds to unaligned EPG-representations.

Observe that the 2-bend number is just the classical bend number.
It is also straightforward to observe that if $d_1 < d_2$, then $b_{d_1}(G) \geq b_{d_2}(G)$ for all graphs $G$. 
Moreover, if there exists $d \in \NN$ such that $b_d(G)=0$, then $b_{d'}(G) = 0$ for all $d' \in \NN$ (as this means that $G$ is an interval graph).

As we have seen in Proposition \ref{lower-slopes}, introducing more slopes may help us reduce the number of bends needed to represent a given graph. Here we show two more examples of this. Consider a {\em wheel graph} 
$W_n$ on $n+1$ vertices ($n \geq 3)$. It follows from the work of Golumbic {\em et al.} \cite{Izraelci} that $W_n$ is not a 1-bend graph (using 2 slopes only) and one can easily find a representation using two bends. On the other hand, for $d \geq 3$, we can represent $W_n$ using 1-bend paths (see Fig. \ref{fig:slopes} (left)). 
Thus $b_2(W_n) = 2$ and $b_d(W_n) =1$ for all $d \geq 3$.

\begin{figure}[h]
\centering
\begin{tikzpicture}[scale=0.5, yscale = 0.9]

\draw (0,0) -- (2,4) -- (4,0);
\draw (0.7,1.6) --++ (-0.9,-1.8) --++ (4,0);
\draw (2.2,-0.1) --++ (2,0) --++ (-1,2);
\draw (4.1,0.4) --++ (-2,4) --++ (-0.7,-1.4);

\draw (0.5,1.4) --++ (0.4,0.8);
\draw (0.7,2) --++ (0.4,0.8);
\draw[densely dashed] (0.9,2.6) --++ (0.4,0.8);
\end{tikzpicture}
\hskip 30pt
\begin{tikzpicture}[scale=0.5]

\draw (0,0) -- (4,0) -- (4,3);
\draw (2,-1) -- (2,2) -- (6,2);

\draw (1,0.1) -- (1.9,0.1) -- (1.9,1);
\draw (2.1,-1) -- (2.1,-0.1) -- (3,-0.1);
\draw (3,2.1) -- (3.9,2.1) -- (3.9,3);
\draw (4.1,1) -- (4.1,1.9) -- (5,1.9);

\end{tikzpicture}
\begin{tikzpicture}[scale=0.5]

\draw (0,0) -- (4,0) -- (4,2);
\draw (1,2) -- (1,-2) -- (5,2);

\draw (0,0.1) -- (0.9,0.1) -- (0.9,1);
\draw (1.1,-1) -- (1.1,-0.1) -- (2,-0.1);
\draw (3.4,0.6) -- (3.9,1.1) -- (3.9,2);
\draw (4.1,0) -- (4.1,0.9) -- (4.6,1.4);

\draw (2.3,-0.1) -- (2.8, -0.1) -- (2.4,-0.5);
\draw (3.5, 0.4) -- (3.2, 0.1) -- (3.6, 0.1);

\end{tikzpicture}
\begin{tikzpicture}[scale=0.5]

\draw (0,0) -- (5,0) -- (1,4);
\draw (2,4) -- (2,-1) -- (5,2);

\draw (1,-0.1) -- (1.9,-0.1) -- (1.9,-1);
\draw (2.1,0.5) -- (2.1,0.1) -- (2.3,0.1);
\draw (2.4,-0.1) -- (2.8, -0.1) -- (2.4,-0.5);
\draw (3.5, 0.4) -- (3.2, 0.1) -- (3.6, 0.1);
\draw (1.4,3.7) -- (1.9,3.2) -- (1.9,4);
\draw (2.1,2) -- (2.1,2.8) -- (2.6,2.3);
\draw (3.5,1.4) -- (3.9,1) -- (3.5,0.6);
\draw (4.5,1.4) -- (4.1,1) -- (4.5,0.6);
\end{tikzpicture}
\caption{\textbf{Left:} Representation of a wheel using 1-bend paths. \textbf{Right:} Representations of $K_{2,s}$ with 1-bend paths.}
\label{fig:slopes}
\end{figure}
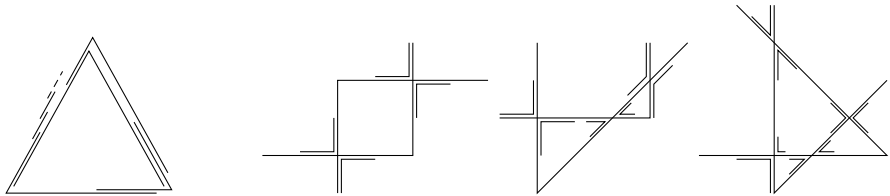

Another examples of graphs with bend number depending on the number of slopes are complete bipartite graphs. Consider e.g. a graph $K_{2,s}$. When only two slopes are available, then $K_{2,s}$ has a 1-bend representation only for $s \leq 4$. Introducing a third slope allows us to represent $K_{2,5}$ and $K_{2,6}$. Fourth slope allows representing $K_{2,7}$ and $K_{2,8}$. By analyzing the possible intersection points of two 1-bend paths, we observe that $K_{2,s}$ for any $s \geq 9$ does not have a 1-bend representation for any number of slopes. On the other hand, every $K_{2,s}$ is a 2-bend graph on two slopes (see Fig. \ref{fig:slopes} (right) and Fig. \ref{aa:dvespojky} (right)).

\section{Conclusions and open problems}  \label{sec:conclusion}
Although all non-trivial classes of EPG-graphs are considered
hard for recognition, not much is known. It is an open problem whether
the recognition problem remains NP-hard for $k$-bend graphs (for $k \geq 3$). 
Note that NP-membership is straightforward, since every $k$-bend graph can be embedded in a $O(n) \times O(n)$ grid \cite{Izraelci,AS}. Thus a polynomial certificate will be the set of $n$ at most $(k+2)$-element sequences of pairs of integers, representing the coordinates of endpoints and bend points of each segment.

\begin{problem}
Is the recognition of $k$-bend graphs NP-complete for every fixed $k \geq 1$?
\end{problem}

For unaligned bend graphs and aligned bend graphs, using more than 2 slopes, naturally arises the question on inclusions between different classes. Also the complexity of the recognition problem is unknown (for more than one bend, when we restrict the number of slopes).
Note that none of our reductions can be easily used. The unaligned version increases
the number of slopes, while in the aligned version a new slope
introduces a new  ``truth value'', but in a way that does not seem to be suitable
for a reduction from any form of coloring.



The {\sc Clique} problem in intersection graphs of geometric objects is related to Helly numbers, which give the minimum number of objects, whose pairwise intersection implies that the intersection of all objects is non-empty (see \cite{Izraelci}).
We may use this approach to restrict the number of places to inspect when finding the maximum clique.
\begin{problem}
Given $k$, what is the maximum number $c(k)$ such that every maximal clique
in $k$-bend graphs is contained in the union of $c(k)$ edge-cliques?
\end{problem} 
In particular, recall from Section \ref{sec:defs} that a clique in a 2-bend graph $G$ may not be contained in the union of two edge-cliques.  It would be interesting to know if three edge-cliques are always enough.

As mentioned before, the {\sc Clique} problem is polynomially solvable in 1-bend graphs. On the other hand, the problem is shown to be NP-complete in {\em 2-interval} graphs~\cite{FGO}, i.e., graphs admitting an intersection model, in which each vertex is represented by two segments in a real line.
Since every 2-interval graph is a 3-bend graph and also a 2-bend graph with 3 slopes, we know that the problem is NP-complete is these classes as well. The complexity for 2-bend graphs remains open. We conjecture it is NP-hard.

\begin{problem}
What is the complexity of the {\sc Clique} problem in 2-bend graphs?
\end{problem}

On the other hand, every $k$-bend graph is also a $(k+1)$-interval graph \cite{Nemci}).
Bar-Yehuda {\em et al.} \cite{BHNSS} show a $2t$-approximation algorithm for the {\sc Maximum Weighted Independent Set} for $t$-interval graphs, which gives a 4-approximation for 1-bend graphs (this bound was later obtained independently by Epstein {\em et al.} \cite{EGM} for the unweighted problem). It would be interesting to improve the approximation ratio or prove a lower bound  (under some well-established complexity assumption).
\begin{problem}
Is it possible to approximate the {\sc Independent Set} problem in 1-bend graphs with a factor $4-\epsilon$ for some $\epsilon > 0$?
\end{problem}

It is not hard to observe that for any two sets $D,D'$ with $|D|=|D'|=3$, one can transform an $EPG(D)$-representation of any graph $G$ to its $EPG(D')$-representation. However, it is not clear if the same holds for sets with at least 4 direction of slopes. It is worth mentioning that there are infinitely many classes of intersection graphs of segments, each of which is parallel to one of 4 slopes (see \v{C}ern\'y {\em et al.} \cite{Slopes}). 

\begin{problem}
Is the minimum number of bends (per path) in an $EPG(D)$-representation of a graph $G$ always equal to $b_d(G)$, for any set $D$ of $d>3$ slopes?
\end{problem}

Our generalization rises yet further questions. Especially, we may put individual vertices into points with integral coordinates. Now, we may ask, how large grid is necessary and sufficient to represent any graph with $n$ vertices and prescribed number of permitted slopes, or even, with prescribed slopes.
Obviously, if we prescribe too steep (but neither vertical nor horizontal) slopes, the grid-size becomes huge (even for graphs of constant size like, e.g., $K_2$).

\bigskip
\noindent \textbf{Acknowledgements.} The authors are sincerely grateful to Katarzyna \.{Z}uk for stimulating discussions about the topic.

\end{document}